\documentclass[journal]{IEEEtran}
\usepackage{todonotes}
\usepackage[final]{changes}
\usepackage{cite}

\ifCLASSINFOpdf
\else
\fi
\usepackage{amsmath}
\ifCLASSOPTIONcompsoc
 \usepackage[caption=false,font=normalsize,labelfont=sf,textfont=sf]{subfig}
\else
 \usepackage[caption=false,font=footnotesize]{subfig}
\fi
\usepackage{graphicx}
\usepackage{amsfonts}
\usepackage{hyperref}
\hypersetup{
    colorlinks=true,
    linkcolor=blue,
    filecolor=magenta,      
    urlcolor=cyan,
    pdftitle={Overleaf Example},
    pdfpagemode=FullScreen,
    }

\usepackage{amsmath}
\usepackage{amssymb}
\usepackage{amsfonts}
\usepackage{mathtools}
\usepackage{amsthm}
\usepackage{multicol}

\newtheorem{theorem}{Theorem}
\newtheorem*{theorem*}{Theorem}
\newtheorem{prop}{Proposition}

\newtheorem{corollary}{Corollary}[prop]
\theoremstyle{definition}
\newtheorem{definition}{Definition}

\begin{document}

\title{Hilbert Transform on Graphs: Let There Be Phase}

\author{Chun Hei Michael~Chan,~\IEEEmembership{Student Member,~IEEE,} Alexandre~Cionca,~\IEEEmembership{Student Member,~IEEE,} ~Dimitri~Van~De~Ville,~\IEEEmembership{Fellow,~IEEE}
        
\thanks{CHM. Chan, A. Cionca and D. Van De Ville are with the Neuro-X Institute, Ecole polytechnique fédérale de Lausanne, and the Department of Radiology and Medical Informatics, University of Geneva, Switzerland. E-mail: chunheimichael.chan@epfl.ch}%
\thanks{This work was supported in part by the Swiss National Science Foundation, Sinergia project ``Precision mapping of electrical brain network dynamics with application to epilepsy'', Grant number 209470.}%
\thanks{Manuscript received XXX; revised XXX.}}

\markboth{Accepted to IEEE Signal Processing Letters (SPL-42226-2025)}%
{}

\maketitle

\begin{abstract}
In the past years, many signal processing operations have been successfully adapted to the graph setting. One elegant and effective approach is to exploit the eigendecomposition of a graph shift operator (GSO), such as the adjacency or Laplacian operator, to define a graph Fourier transform for projecting graph signals on the corresponding basis. However, the extension of this scheme to directed graphs is challenging since the associated GSO is non-symmetric and, in general, not diagonalizable. Here, we build upon a recent framework that adds a minimal number of edges to allow diagonalization of the GSO and thus provide a proper graph Fourier transform. 
Furthermore, we show that such minimal addition of edges creates a cycle cover and that it is essential for the phase analysis of a signal throughout the graph. Concurrently, we propose a generalization of the Hilbert transform interpreted over the newfound cycle cover, which re-establishes intuitions from traditional Hilbert transform, equivalent to the generalized Hilbert transform on a single cycle. This generalization leads to a number of simple and elegant recipes to effectively exploit the phase information of graph signals provided by the graph Fourier transform. The feasibility of the approach is demonstrated on several examples. 
\end{abstract}

\begin{IEEEkeywords}
Graph signal processing, Harmonic analysis, Hilbert transform, Directed graph.
\end{IEEEkeywords}

\IEEEpeerreviewmaketitle

\section{Introduction}
\IEEEPARstart{G}{raph} signal processing (GSP) has emerged as an attractive new methodology for analyzing structured data \cite{shuman_emerging_2013, ortega_graph_2018, marques_signal_2020}. In most cases, an undirected graph is considered, which implies a symmetric adjacency matrix. The definition of the graph Fourier transform (GFT) by the eigendecomposition of a graph shift operator (GSO)~\cite{sandryhaila_discrete_2013}, commonly the adjacency or Laplacian matrix, leads then to real-valued eigenvalues and eigenvectors, and constitutes a unitary transform suitable for the analysis and synthesis of graph signals. Different operations, such as filtering, can then be defined and implemented in the spectral graph domain \cite{sandryhaila_big_2014, chen_signal_2014, xuesong_shi_infinite_2015}. 

For directed graphs, however, the eigendecomposition of the GSO leads to complex-valued eigenvalues and eigenvectors. Moreover the GSO itself becomes non-symmetric and is in general non-diagonalizable. While the Jordan normal form (JNF) has been suggested for providing the spectral decomposition of the GSO~\cite{sandryhaila_discrete_2014}, this comes with several complications. First, the JNF is numerically unstable \cite{cox_computational_1990}, especially for larger graphs. Second, non-trivial Jordan blocks correspond to spectral components with dimension larger than one, which impede GSP operations such as filtering and sampling \cite{shi_graph_2019}. 

Following these two main problems in the numerical application of the JNF in GSP, alternatives have been developed to ensure orthogonality through optimized bases \cite{sardellitti_graph_2017, shafipour_directed_2019}. Both these works introduce directed total variation measures and yield real bases. Moreover, \cite{shafipour_directed_2019} finds bases with uniformly sampled frequency. However, both methods only work on real signals and bypass the GSO, meaning that only spectral operations would be allowed. Another approach proposed is to once again bypass the standard GSO and instead use the diagonalizable Hermitian Laplacian \cite{brefeld_graph_2020}, which leads to a basis that satisfies Parseval property. These alternatives would fulfill properties such as orthogonality, evenly spread frequencies, and most importantly, diagonizability, but turn out to yield real eigenvalues, preventing the immediate definition of phase defined from complex values. 

\added{Recent work allows to obtain both properties of diagonalizability and complex-valued eigendecomposition~\cite{seifert_digraph_2021}. Specifically, by leveraging perturbation theory \cite{moro_low_2003, savchenko_change_2004} on the Jordan blocks, they propose to dismantle non-trivial Jordan blocks that are rank-2 or above, by adding a minimal number of edges to the directed graph. This method enables conventional diagonalization of the modified graph adjacency and thus permits the definition of the GFT as the projection on the dual basis derived from the GSO eigendecomposition. This method also yields a minimally different (spectrally) approximation of the original adjacency matrix \cite{seifert_digraph_2021}. The new adjacency remains a real-valued matrix that has an eigendecomposition with eigenvalues that are either real-valued or pairs of complex conjugates \cite{horn_matrix_1985}. A parallel can thus be drawn between traditional discrete Fourier transform's (DFT) pairs of conjugate frequencies and the pairs of conjugate eigenvalues found in the presented eigendecomposition of the directed graph. A direct application of this property is the Hilbert transform.}

\deleted{Hence in this work we start by considering, as GSO the adjacency matrix \mbox{\cite{sandryhaila_discrete_2013}} representation of the directed graph, which exhibits complex-valued eigenvalues. Secondly, we use recent work \mbox{\cite{seifert_digraph_2021}}, leveraging perturbation theory of Jordan structures \mbox{\cite{moro_low_2003, savchenko_change_2004}}, that proposed to dismantle non-trivial Jordan blocks that is rank-2 or above Jordan blocks while adding a minimal amount of edges to the directed graph. This method enables conventional diagonalization of the modified graph adjacency, and thus permits defining the GFT as the projection on the dual basis derived from the GSO eigendecomposition. This method also proves to yield a minimally different (spectrally) approximation of the original adjacency \mbox{\cite{seifert_digraph_2021}}. The newly generated adjacency remains a real-valued matrix that brings an eigendecomposition with eigenvalues that are either real-valued or pairs of complex conjugates \mbox{\cite{horn_matrix_1985}}. Thus a parallel can be drawn between traditional discrete Fourier transform's (DFT) pairs of conjugate frequencies and the pairs of conjugate eigenvalues found in the presented eigendecomposition of the diagonalizable directed graph. A direct application of this property is the Hilbert transform.}

Prior work from \cite{venkitaraman_hilbert_2019} defines the graph Hilbert transform (GHT) through JNF, and while they show its equivalence to the traditional Hilbert transform, the authors revert to examples where diagonalizability and invertibility are satisfied, thus not giving a method to apply GHT on general graphs. On a non-diagonalizable graph, as we will see, the GHT through JNF does not provide phase information on all parts of the graph signal.

Instead, here we propose using the mentioned minimal graph perturbation to render the GSO diagonalizable and invertible. While \cite{seifert_digraph_2021} only states that this minimal perturbation add cyclicity in the graph, we prove that it creates a cycle cover. Furthermore, we show through a simple example that such a cycle cover is essential for phase analysis throughout the graph. Leveraging this newfound cycle cover, we introduce a more intuitive and familiar definition of the GHT, providing interpretation for instantaneous phase and amplitude on cycles as in traditional HT. Finally we highlight its properties and specifically the extension to the graph analytical signal by insightful examples.

\section{Background}
In what follows, we denote a matrix $\bf X$ as bold uppercase and a vector $\bf x$ as bold lowercase. An element $k$ of a vector will be indicated as ${\bf x}[k]$ and a series of vectors ${\bf x}_m$ will be indexed using subscript. The conjugate will be written as $\cdot^\star$.

Let us consider a directed graph $\mathcal G=(V,E)$, with $V$ being the set of nodes and $E$ the set of edges. The directed graph $\mathcal G$ has $N$ nodes and is also characterized by the (non-symmetric) adjacency matrix ${\bf A}$. In \cite{seifert_digraph_2021}, the perturbed GSO is obtained by adding iteratively rank-1 matrices (i.e., single edges) to $\bf A$. On each iteration, the largest Jordan block is identified by considering the group of the largest number of collinear Jordan eigenvectors. Using \cite[Th.~4, Lem.~12]{seifert_digraph_2021} a candidate edge to dismantle a Jordan block is an edge $(n,m)$  such that $\left|{\bf t}[n]\right| \left|{\bf q}[m]\right|\neq 0$, with $\bf{t}$ and $\bf{q}$ one of the left and right eigenvectors associated to the largest Jordan block, respectively. The same criterion is used on eigenvectors associated to the eigenvalue 0. At the end, the process yields a graph $\mathcal G'=(V,E')$ with associated adjacency ${\bf A}'$ that is guaranteed to be diagonalizable and invertible.

\section{From Graph Phase to Hilbert Transform}
\subsection{Phase in the Graph Fourier Domain}
Let us consider the eigendecomposition of the perturbed adjacency:
$$
  {\bf A'}={\bf U}{\bf \Lambda}{\bf U}^{-1},
$$
where the eigenvalues ${\bf \Lambda}[k,k]=\lambda_k$ are either real-valued with corresponding real-valued eigenvectors ${\bf u}_k$, $k=1,\ldots,N$, that are the columns of ${\bf U}$; or complex-valued, and then come in pairs of complex conjugate eigenvalues and eigenvectors\cite{horn_matrix_1985}. For instance, for a pair of complex-valued eigenvalues with indexes $(k_1,k_2)$, we have $\lambda_{k_1}=\lambda_{k_2}\! ^\star$ with eigenvectors ${\bf u}_{k_1}={\bf u}_{k_2}\! ^\star$. Consequently, the GFT $\hat{\bf x}={\bf U}^{-1}{\bf x}$ of a real-valued graph signal ${\bf x} \in {\mathbb R}^N$ will also have complex conjugate coefficients for these pairs; i.e., $\hat{\bf x}[k_1]=\hat{\bf x}[k_2]^\star$. 

\subsection{Graph Hilbert Transform and Analytical Graph Signal}
We introduce the GHT by defining the following filter in the spectral domain by the diagonal matrix $\hat{\bf H}$: 
\begin{equation}
    \hat{\bf H}[k,k] = \left\{ \begin{array}{ll} -j, & \text{imag}(\lambda_k)>0,\\ +j, & \text{imag}(\lambda_k)<0,\\ 0, & \text{imag}(\lambda_k)=0. \end{array} \right.
\end{equation}
Thus GFT coefficients corresponding to complex-valued adjacency eigenvalues are multiplied with $-j$ and $+j$ according to the sign of the imaginary part of the eigenvalues. A positive/negative imaginary part of the eigenvalue can be interpreted as a positive/negative frequency. As aforementioned, GFT coefficients of a real-valued graph signal are complex conjugate for the pairs of complex eigenvalues. Pairs of conjugate coefficients, when multiplied by $-j$ and $+j$, preserve conjugate relationship leading to a real synthesized signal. Therefore, the Hilbert transform ${\mathcal H({\bf x})}={\bf U}\hat{\bf H}{\bf U}^{-1}{\bf x}$ of such a graph signal is also real-valued.
Subsequently, we can define the equivalent of the analytical signal in traditional Hilbert transform \cite{hahn_hilbert_1996} to graph, namely the analytical graph signal: 
\begin{equation}
  \tilde{\bf x} = {\bf x}+j{\mathcal H({\bf x})}={\bf U} ({\bf I}+j\hat{\bf H})\hat{\bf x}, 
\end{equation}
which can be represented in terms of instantaneous amplitudes and phases: 
\begin{eqnarray*}
    {\mathcal A}({\bf x})[k] & = & \left\|\tilde{\bf x}[k]\right\|, \\
    \varphi({\bf x})[k] & = & \arctan\left( \text{Im}(\tilde{\bf x}[k]), \text{Re}(\tilde{\bf x}[k]) \right).
\end{eqnarray*}
The analytical signal is employed in amplitude and frequency modulation (AM/FM) \cite{proakis_digital_2008}. In a similar fashion, the analytical graph signal sets the ground for graph-wide amplitude and frequency modulation.

\subsection{Interpretation of Graph Hilbert Transform}
In traditional signal processing, the Hilbert transform of a signal provides the magnitude of its envelope and phase of its oscillatory pattern. This interpretation directly applies to a directed cycle graph, which is a model for a 1-D discrete signal with periodic boundary conditions. For a more general graph, we need to acknowledge its structure in terms of subcycles. The described method of minimal perturbation allows us to infer the existence of these subcycles, where each one satisfies periodic boundary conditions. With the following Proposition and Corollary, we put forward a generalized interpretation of the GHT instantaneous amplitude and phase.

\begin{figure}
    \centering
    \captionsetup{justification=centering}
  \subfloat[\label{2a}]{%
        \includegraphics[width=0.39\linewidth]{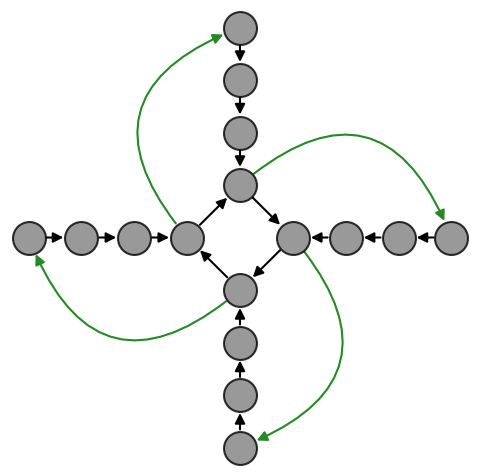}}
  \subfloat[\label{2b}]{%
        \includegraphics[width=0.39\linewidth]{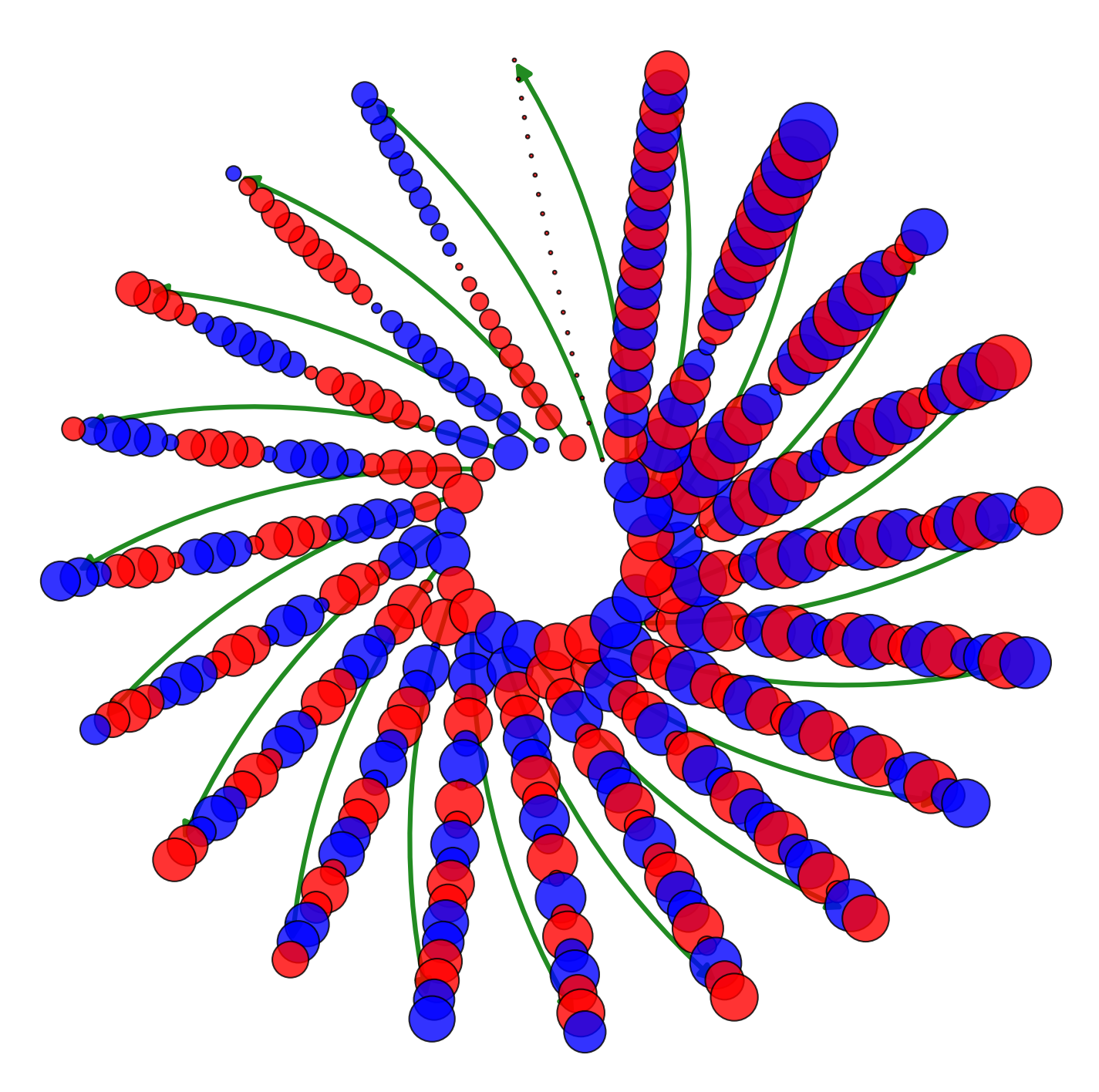}}\\
    \subfloat[\label{cyclic-ght}]{%
       \includegraphics[width=0.39\linewidth]{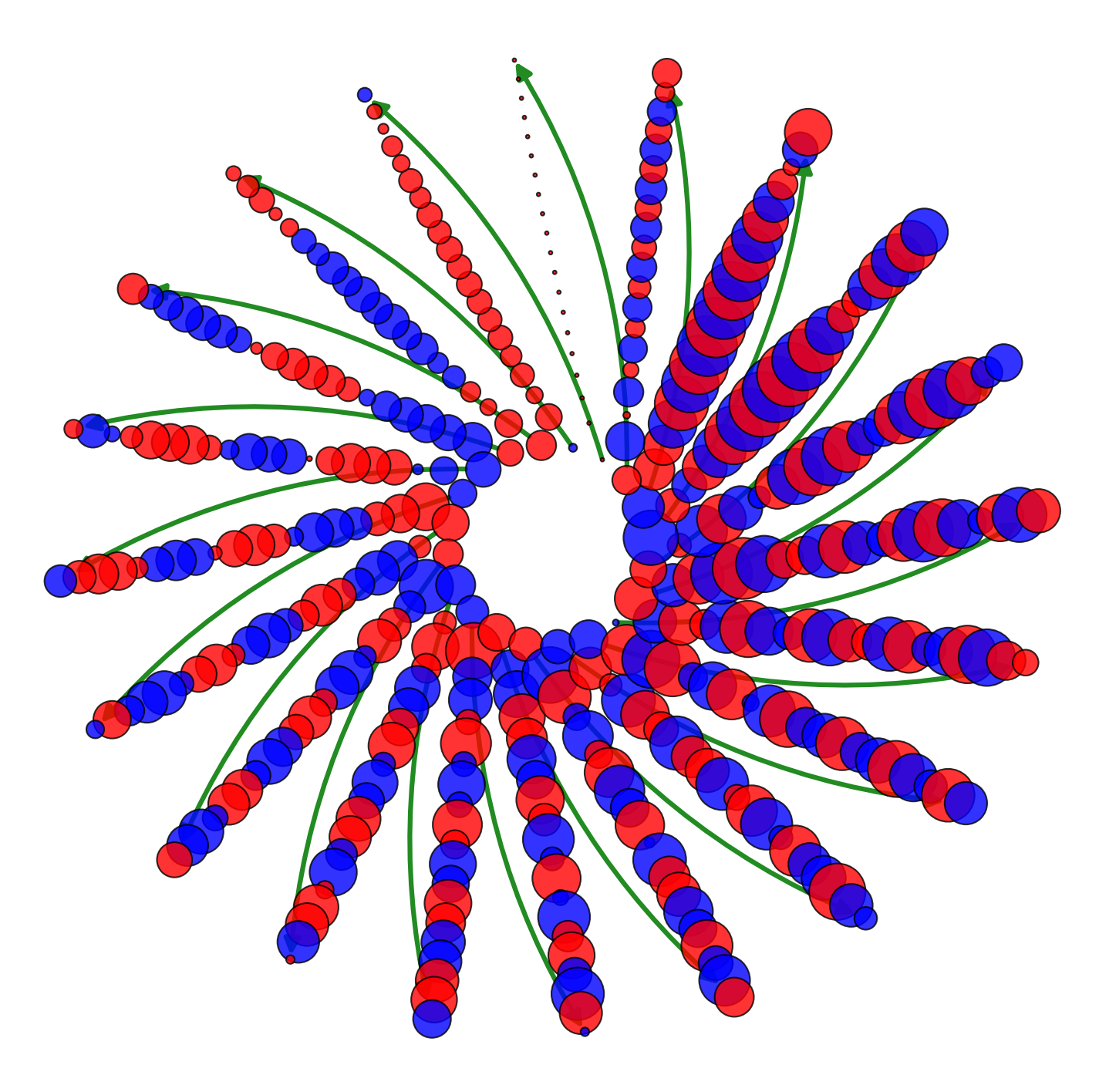}}
    \subfloat[\label{jordan-ght}]{%
       \includegraphics[width=0.39\linewidth]{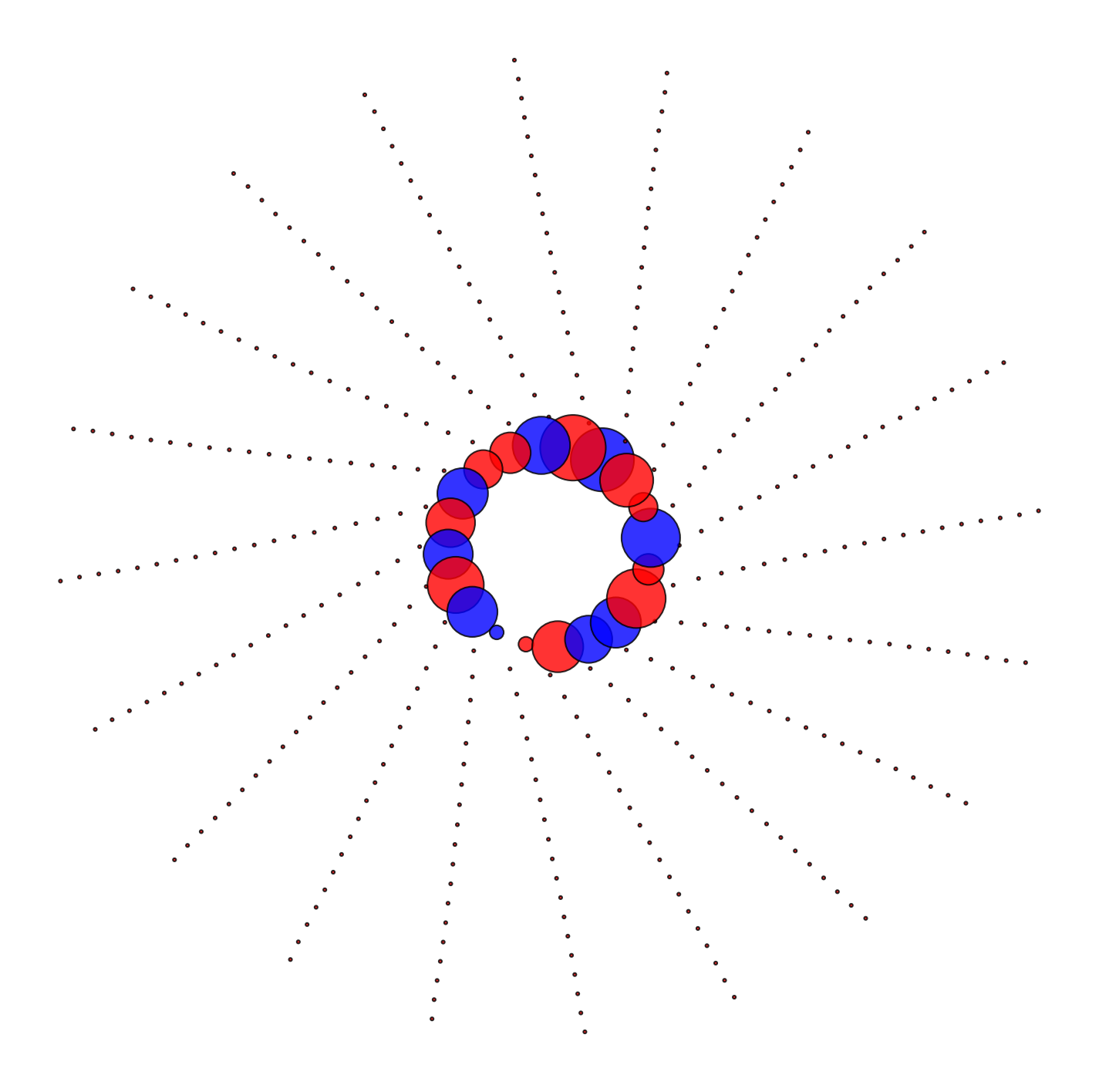}}
  \caption{(a) Schematic of synthetic graph. The green edges indicate those that are added to render the graph's adjacency matrix diagonalizable and invertible. (b) Synthetic graph and graph signal. Again, the green edges are the ones of the modified graph. (c) Proposed GHT applied to the graph signal of (b). (d) GHT using JNF applied to the graph without green edges.}
\end{figure} Let us recall that a graph $\mathcal G$ is $r$-acyclic if any collection of vertex-disjoint cycles of $\mathcal G$ covers at most $N-r$ nodes. A cycle cover is a collection of cycles of $\mathcal G$ that includes all nodes of $\mathcal G$. We also bring up the following definition and theorem \cite{li_connections_2023}.
\added{
\begin{definition} \cite{li_connections_2023} The set $S_{\mathcal{G}}$ of matrices associated to subgraphs of $\mathcal{G}$: 
    $S_{\mathcal{G}}=\{A_{sub} \text{ s.t }\mathcal{G}_{sub}=(V,E_{sub}) \text{ and } E_{sub} \subseteq E\}$.
\end{definition}}

\begin{theorem}\label{T1} \cite[Th.~4.4]{li_connections_2023}
    A graph $\mathcal{G}$ is $r$-acyclic if and only if every matrix $A_{sub}\in S_{\mathcal{G}}$ has at least $r$ zero eigenvalues. 
\end{theorem}
\begin{prop} \label{P1}
$\mathcal G'$ admits a cycle cover; i.e., it can be decomposed into $M$ subcycles $\mathcal C_1, \ldots, \mathcal C_M$ such that $\bigcup_{m=1}^M \mathcal C_m=V$.
\end{prop}

\begin{proof}
Let us first show that the graph $\mathcal{G}'$ is $0$-acyclic. By contradiction, let us assume that our graph $\mathcal{G}'$ is $r$-acyclic with $r\neq0$. By using \added{Theorem.~\ref{T1}} in the direct sense, we have that every matrix $A_{sub}\in S_{\mathcal{G}'}$ has at least one zero eigenvalue. However, the associated adjacency matrix $A'\in  S_{\mathcal{G}'}$ has no zero eigenvalues due to invertibility. This leads to a contradiction and, therefore, $\mathcal{G}'$ can only be $0$-acyclic, consequently, $\mathcal{G}'$ admits a cycle cover.
\end{proof}

\begin{corollary}
Given graph signals ${\bf x}_m\in\mathbb{R}^N$ supported on subcycles $\mathcal C_m$, $m=1,\ldots,M$, (i.e., ${\bf x}_m[k]=0$ for $k\notin \mathcal{C}_m$), the following properties of the Hilbert amplitude and phase of the aggregated graph signal ${\bf x}=\sum_{m=1}^M {\bf x}_m$ hold:
\begin{align*}
    {{\mathcal A({\bf x})}[k]^2} &= \sum_{m=1}^M {\mathcal A({\bf x}_m)}[k]^2 \\
    &+ \left.\sum_{m\neq n}^M {\bf x}_m[k] {\bf x}_n[k] + {\mathcal H}({\bf x}_m)[k] {\mathcal H}({\bf x}_n)[k]\right., \\
{\bf {\varphi(x)}}[k] &= \arctan\left(\sum_{m=1}^M {\mathcal H}({\bf x}_m)[k], {\sum_{m=1}^M {\bf x}_m[k]}\right).
\end{align*}
\added{
\begin{proof}
    The proof of this Corollary can be found in the Supplementary Material.
\end{proof}}
\end{corollary}

These properties explain how the signals from different cycles combine for overlapping nodes; i.e., nodes of overlapping subcycles. For a node $k$ that exclusively belongs to one subcycle $\mathcal C_m$, we have ${\bf x}[k]={\bf x}_m[k]$ and thus the equations revert to the conventional interpretation of instantaneous amplitude and phase in that cycle. For nodes of overlapping subcycles, the strengths of the contributing signals ${\bf x}_m$ will determine their contribution. 

\section{Experimental Results}
\subsection{Setup}

\begin{figure}
    \centering
    \captionsetup{justification=centering}

    \subfloat[Per fan inst.\ amplitude\label{2e}]{%
       \includegraphics[width=0.4\linewidth]{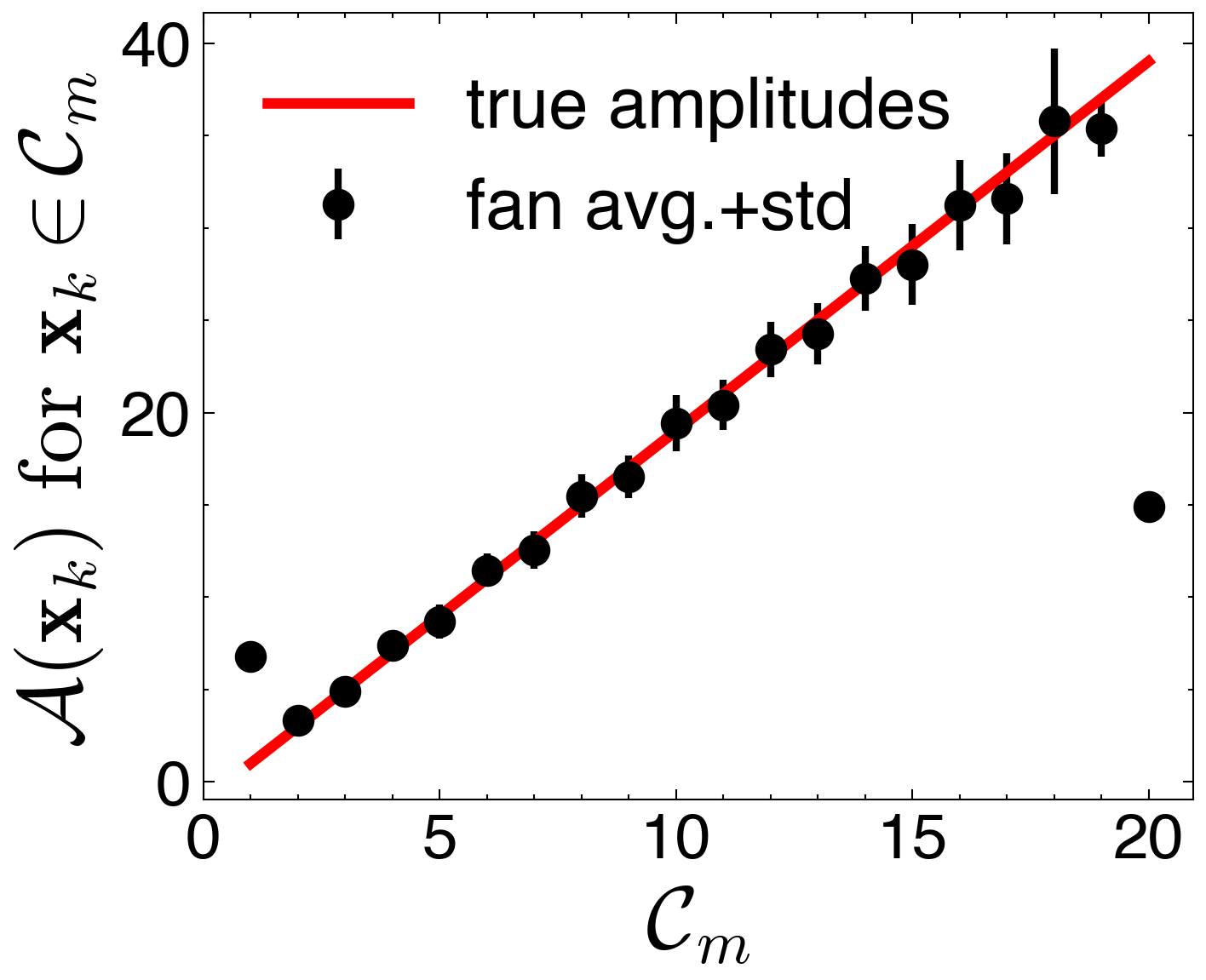}}
  \subfloat[Per fan inst.\ frequency\label{2f}]{%
\includegraphics[width=0.39\linewidth]{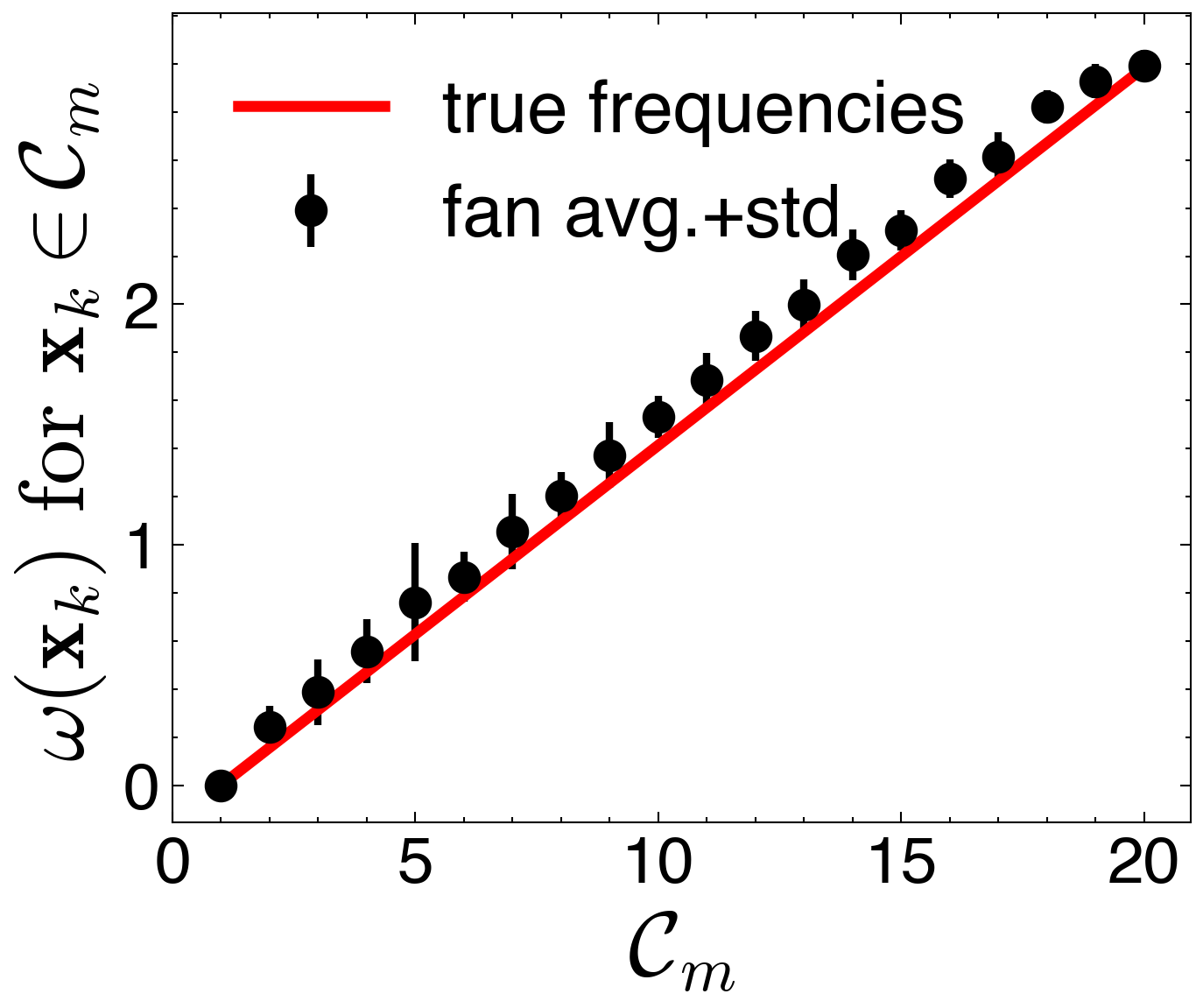}}\\
          \subfloat[Central cycle\label{2d}]{%
        \includegraphics[width=0.9\linewidth]{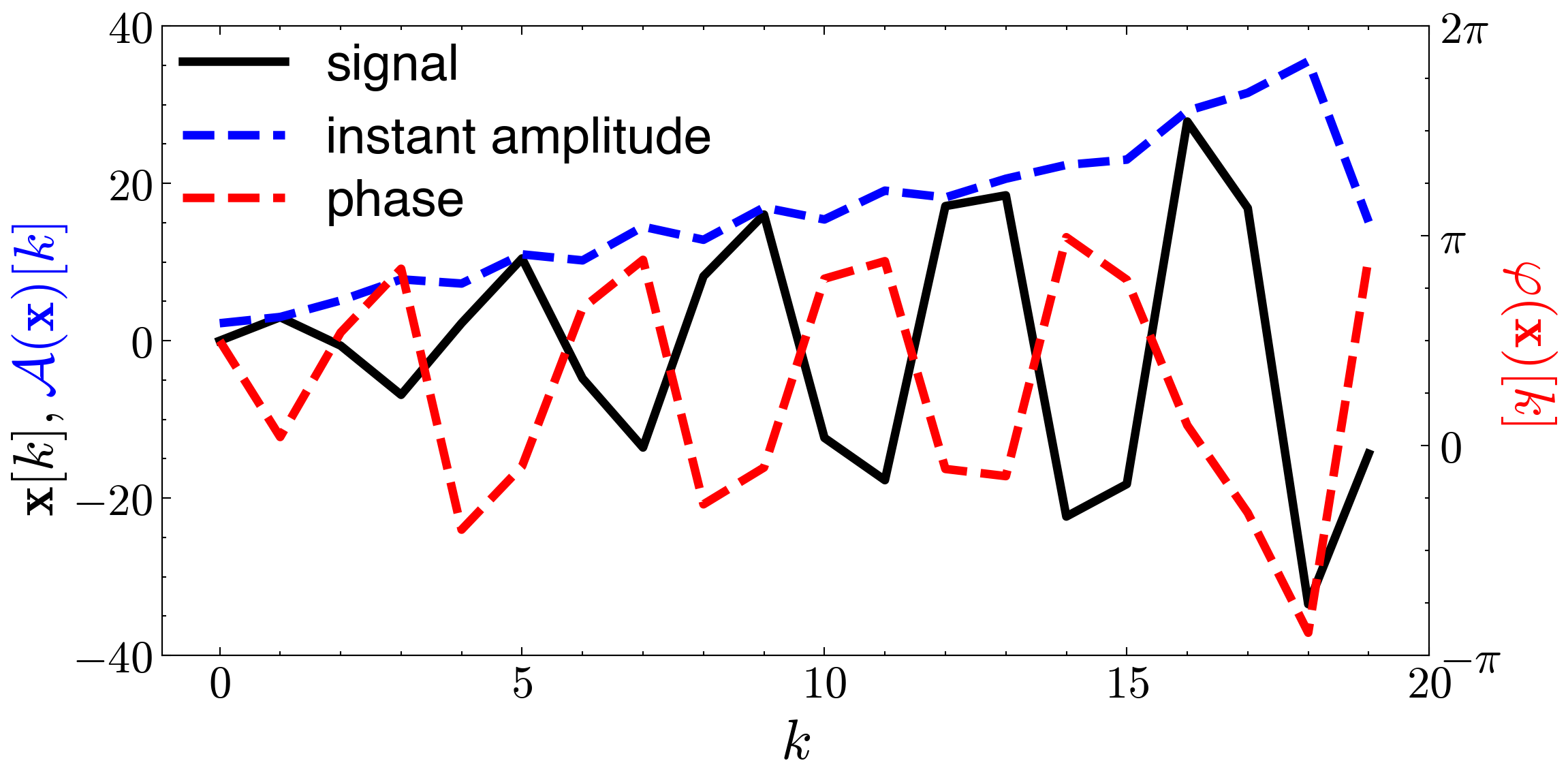}}
  \caption{(a) Instantaneous amplitude distribution across fans $\mathcal{C}_m$. (b) Instantaneous frequency distribution across fans $\mathcal{C}_m$. (c) Graph signal, instantaneous amplitude, and phase on the central cycle.}
\end{figure}

\subsubsection{Synthetic Data} The structure of our first graph is illustrated in Fig.~\ref{2a} consisting of a central directed cycle graph with $N_C$ nodes that each have an outgoing fan containing $N_F$ nodes, for a the total number of nodes of $N_C N_F$. We associate to each outgoing fan $\mathcal{C}_m$, $m=1,\ldots,N_F$ sinusoidal signals:
\begin{equation*}
   {\bf x}_m[k]=(2m+1) \sin\left(\frac{m+1}{2N_f}2\pi k + \frac{5N_c+4}{20N_c}2\pi m\right), 
\end{equation*}
which manifest increasing frequency and amplitude on the fans, and at the same time a signal with fixed frequency on the central cycle.

\begin{figure*}
    \centering
    \captionsetup{justification=centering}
    \subfloat[Manhattan graph \label{3a}]{%
       \includegraphics[width=0.17\linewidth]{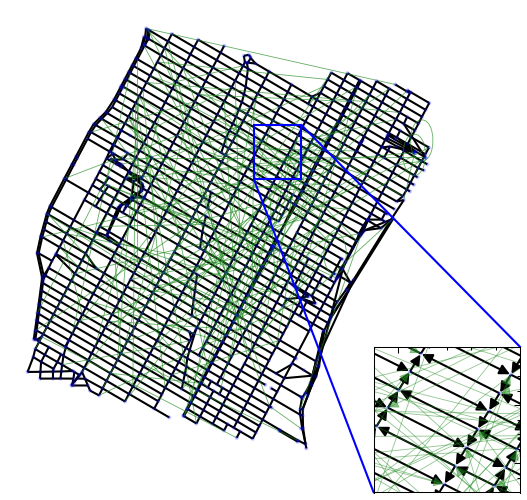}}
    \subfloat[Regular 2D grid\label{3b}]{%
    \includegraphics[width=0.17\linewidth]{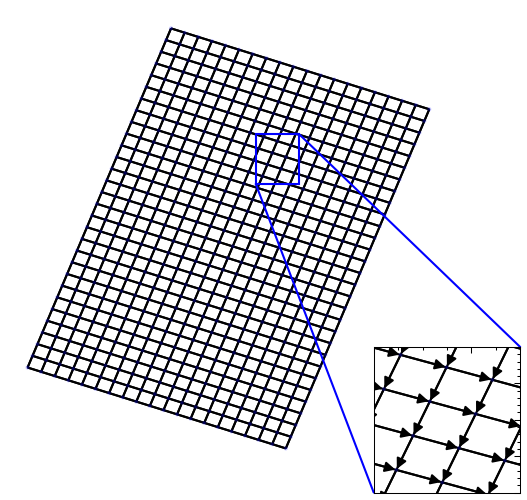}}    
    \subfloat[Graph signal on (b) \label{4d}]{%
       \includegraphics[width=0.16\linewidth]{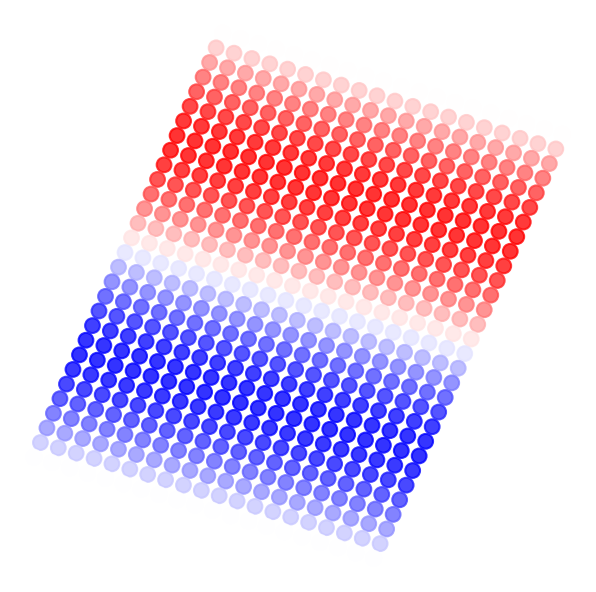}}
    \subfloat[GHT of (c) \label{4b}]{%
       \includegraphics[width=0.16\linewidth]{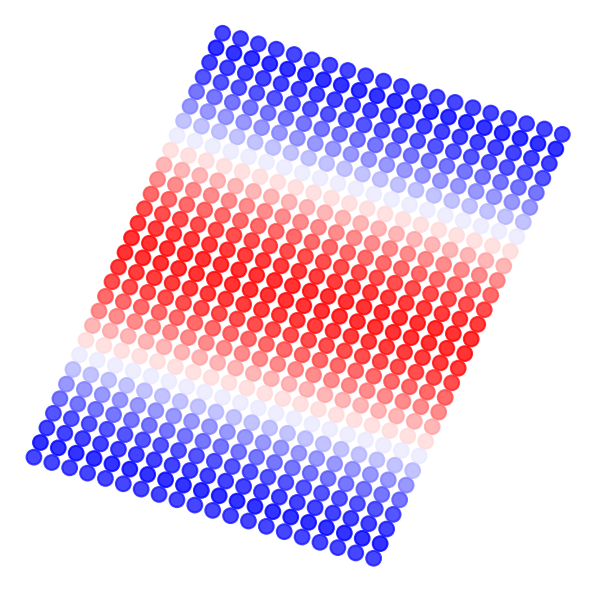}}
  \subfloat[Graph signal on (a)\label{4a}]{%
       \includegraphics[height=2.8cm,width=0.16\linewidth]{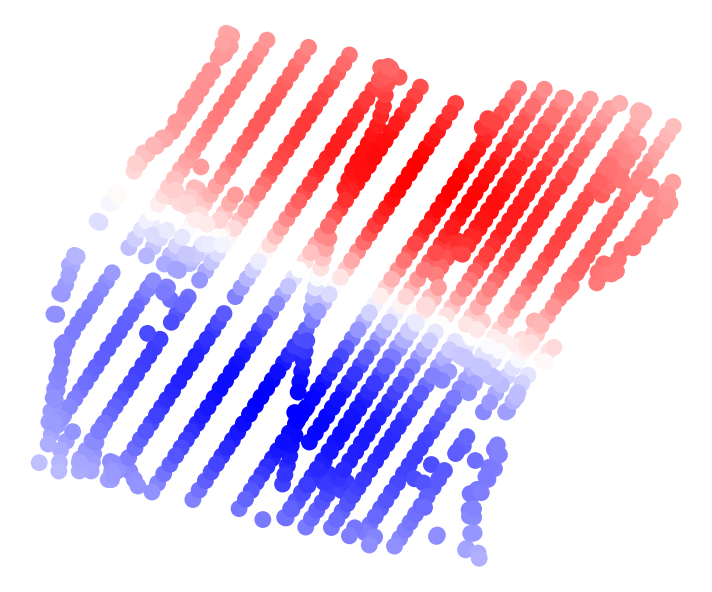}}
    \subfloat[GHT of (e)\label{4c}]{%
       \includegraphics[height=2.8cm,width=0.16\linewidth]{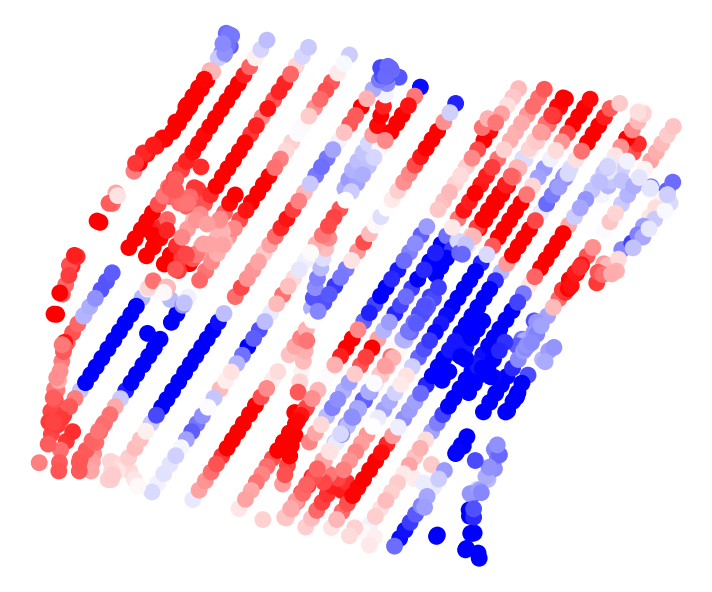}}
  \caption{\label{fig4} (a) Manhattan graph. The green edges indicate those that are added by the perturbation. (b) Regular 2D grid with periodic boundary conditions along obliques. (c,e) Graph signal is a planar wave oriented along the same direction as avenues in Fig.~\ref{4a}. (d) On the regular 2D grid, the GHT of the graph signal shows a clear $\pi/2$ phase shift along the propagation direction. (f) On the Manhattan midtown graph, the GHT of the graph signal reveals the effect of the underlying directed graph in phase shifting.}
\end{figure*}

\subsubsection{Experimental Data} From the Manhattan graph \cite{domingos_graph_2020}, we take the midtown subgraph composed of many one-way avenues and streets (in both directions), see Fig.~\ref{3a}. We also created a graph from a regular 2D grid of similar size where the nodes are connected from top-to-bottom and left-to-right with periodic boundary conditions, see Fig.~\ref{3b}. As a graph signal, we created a planar wave pattern along the direction of the avenues (i.e., approximately north-south). The wave is that of a sinusoid and the frequency of the wave is chosen such that one period covers the midtown area, this yields a sinusoidal plane wave. The graph signal is illustrated in Fig.~\ref{4d},~\ref{4a} on the regular 2D grid and midtown Manhattan graph, respectively.
The chosen graphs and signals reflect domains and monogenic signals where traditional HT is commonly applied: Fig.~\ref{2a} evokes temporal applications with overlapping nodes, Fig.~\ref{3b} recalls image-based uses via a lattice structure, and the Manhattan graph Fig.~\ref{3a} combines grid-like topology with real graph directionality.
\subsection{Results}

\subsubsection{Synthetic Data} 
We instantiate the graph of Fig.~\ref{2a} for $N_C=20$ fans of length $N_F=20$. The adjacency matrix is neither diagonizable nor invertible, with the Jordan block decomposition providing $N_C$ blocks with eigenvalues zero. Removing the non-trivial Jordan blocks leads to the directed rosace graph as shown schematically in Fig.~\ref{2b}, turning fans into subcycles linked to the central part.  

Now diagonalizable, we can apply the proposed GHT to the signal in Fig.~\ref{2b}. This leads to the expected phase shift on the fans (Fig.~\ref{cyclic-ght}). However, when applying the GHT based on the JNF of the original graph~\cite{venkitaraman_hilbert_2019}, we notice that the operation cancels the signal on the fans and only provides a phase shift on the central cycle (Fig.~\ref{jordan-ght}).

We further computed the instantaneous amplitude and phase using the proposed GHT. In addition, the instantaneous phase was converted into instantaneous frequency by computing the phase-unwrapped derivative along the fan as follows
\begin{equation}
  \omega({\bf x})[k] = \left( \varphi({\bf x})[k] - \varphi({\bf x})[k_{+1}] \right)+2l\pi,
\end{equation}
where $k_{+1}$ indicates the next node on the fan and $l$ is the integer that turns $\omega$ in the interval $]-\pi,\pi]$ following Itoh's unwrapping method~\cite{itoh_analysis_1982}. 
As shown in Figs. \ref{2e} and \ref{2f}, the average amplitude and frequency per fan are accurate estimates of the ground truth expressed in the synthetic graph signal. In addition, the instantaneous amplitude on the central cycle recovers the envelope of the oscillating and increasing signal, and the phase varies constantly since the frequency is fixed on the central cycle (Fig. \ref{2d}).

This example perfectly illustrates Proposition \ref{P1}. First, the GHT is fully in line with the classical Hilbert transform on the fans that form non-overlapping cycles, and thus amplitude, phase, and frequency are perfectly recovered. Second, the GHT has the capability to provide similar properties for an overlapping cycle (i.e., the central cycle of this graph). It is clear that the graph signal expressed on the central cycle is ``compatible'' with the one on the fans. Nevertheless, instantaneous amplitude and phase are in accordance with the meaning of the conventional Hilbert transform.

\subsubsection{Experimental Data} 
The adjustment of the adjacency matrix of the Manhattan midtown graph leads to $252$ extra edges (compared to the $1970$ original ones), indicated in green in Fig.~\ref{3a}. The graph based on the regular 2D grid does not require any adjustment (Fig.~\ref{3b}). First, we look at the effect of the GHT on the planar wave graph signal expressed on the regular grid (Fig.~\ref{4b}). We observe a phase shift of $\pi/2$ along the propagation direction, turning the sine into a cosine. Second, the GHT of the same signal expressed on the Manhattan midtown graph produces a complicated scattering of wave (Fig.~\ref{4c}) due to the intricate pattern of directed edges (by one-way streets and avenues crossing, see Fig.~\ref{3a}). We can recognize different patterns in the east and west of the midtown. The implementation is publicly available\footnote{Repository:  \url{https://github.com/MIPLabCH/Graph-Hilbert-Transform}}. 

\section{Discussion \& Outlook}
While the Hilbert transform has been extended previously to the directed graph setting \cite{venkitaraman_hilbert_2019} through the JNF, we highlight here that it is key to first have the eigendecomposition of the GSO to be well-defined; i.e., diagonalizable and invertible, and, consequently, to have a cycle cover in order to enable phase analysis across the graph. In addition, we provide an interpretation of the GHT by constructing the graph signal from contributions on subcycles. This leads to a useful intuition on how the GHT combines information from different cycles. Our first example illustrates how the GHT matches the conventional Hilbert transform on non-overlapping cycles, and can remain interpretable on overlapping ones. \deleted{For general graphs, such as the Manhattan midtown graph, all nodes are overlapping nodes and directionality is complex.} \replaced{O}{Interestingly, o}{n a graph that corresponds to a regular 2-D grid, the GHT behaves as expected from monogenic signal analysis \cite{felsberg_monogenic_2001, unser_monogenic_2008}; that is, the phase shift of the Hilbert transform occurs along the propagation direction of the planar wave.  However, when applying the GHT on the \added{more general Manhattan }midtown graph \added{where the distribution of directions is non-uniform}, the complexity of the graph structure is reflected in an intricate pattern that is phase-shifted differently in different parts. 

In terms of signal representation for GHT, if priors on subcycles' signal are unavailable, we propose that representation is best achieved through decomposition into a weighted sum of GSO eigenvectors, with cycle decompositions ensuring near-linear instantaneous phase for each eigenvectors, akin to all Fourier bases being of linear phase \cite{oppenheim_signals_1997}. While such decomposition is natural, further study is needed and we believe this merits separate work along a focused investigation on real-world applications.

In this work, we used the adjacency matrix as the GSO. However, the directed Laplacian \cite{singh_graph_2016} is also a valid choice and in particular other novel decomposition could allow for phase definitions, such as the use of the polar decomposition, by looking into the introduced unitary matrix $\bf Q$ \cite{kwak_frequency_2024} and the Schur decomposition which exhibits real and imaginary valued blocks in its upper quasi-triangular matrix \cite{xiao_joint_2023}.

Future work includes extending the known applications of the HT \cite{taner_complex_1979,bruns_fourier-_2004,proakis_digital_2008} to directed graphs. \added{In particular the GHT can analyze amplitude or frequency modulation in graph signals, such as synchrony of wave-like patterns encountered in neuronal recordings on irregular domains~\cite{le_van_quyen_comparison_2001}.}
\ifCLASSOPTIONcaptionsoff
  \newpage
\fi

\bibliographystyle{ieeetr}

\bibliography{references}

\end{document}